\documentclass[a4paper,reprint,onecolumn,notitlepage,aip,nofootinbib]{revtex4-1}
\usepackage[left=1.5cm,right=1.5cm,top=1cm,bottom=1.5cm,includeheadfoot]{geometry}

\usepackage[english]{babel}
\usepackage[T1]{fontenc}
\usepackage{tgtermes} 
\usepackage{tgheros} 

\usepackage{amssymb}
\usepackage{amsmath}
\usepackage{amsthm}
\usepackage{mathtools}
\usepackage{braket}
\usepackage{mathrsfs} 
\usepackage{enumerate}
\usepackage{booktabs}
\usepackage[dvipsnames]{xcolor}
\usepackage[shortlabels]{enumitem}
\usepackage[normalem]{ulem}
\PassOptionsToPackage{hyphens}{url} 
\usepackage[breaklinks=true]{hyperref}
\bibpunct{[}{]}{;}{n}{}{} 

\usepackage{tikz}
\usetikzlibrary{patterns}

\theoremstyle{plain}
\newtheorem{theorem}{Theorem}
\newtheorem{proposition}[theorem]{Proposition}
\newtheorem{lemma}[theorem]{Lemma}

\theoremstyle{remark}
\newtheorem{remark}{Remark}

\newcommand{\R}{\mathbb{R}}
\newcommand{\C}{\mathbb{C}}
\newcommand{\N}{\mathbb{N}}
\newcommand{\Z}{\mathbb{Z}}
\newcommand{\T}{\mathbb{T}}

\renewcommand{\d}{\mathrm{d}}
\newcommand{\e}{\mathrm{e}}
\renewcommand{\i}{\mathrm{i}}

\renewcommand{\H}{\mathcal{H}}

\DeclareMathOperator{\Tr}{Tr}

\newcommand{\densset}{\mathscr{I}}
\newcommand{\psiset}{\mathscr{W}}
\newcommand{\affspace}{\mathscr{X}}
\newcommand{\vrepset}{\affspace_{>0}}
\newcommand{\DMset}{\mathscr{D}}

\DeclarePairedDelimiter{\abs}{\lvert}{\rvert}
\DeclarePairedDelimiter{\norm}{\Vert}{\rVert}
\DeclarePairedDelimiter{\innerproduct}{\langle}{\rangle}

\begin{document}
\title{\texorpdfstring{$v$}{v}-representability on a one-dimensional torus at elevated temperatures}

\author{Sarina M. Sutter}
\affiliation{Theoretical Chemistry, Faculty of Exact Sciences, VU University, Amsterdam, The Netherlands}

\author{Markus Penz}
\email[Electronic address:\;]{m.penz@inter.at}
\affiliation{Arnold Sommerfeld Centre for Theoretical Physics,
Ludwig-Maximilians-Universität München, Munich, Germany}
\affiliation{Munich Center for Quantum Science and Technology, Munich, Germany}
\affiliation{Department of Computer Science, Oslo Metropolitan University, Oslo, Norway}
\affiliation{Max Planck Institute for the Structure and Dynamics of Matter and Center for Free-Electron Laser Science, Hamburg, Germany}

\author{Michael Ruggenthaler}
\affiliation{Max Planck Institute for the Structure and Dynamics of Matter and Center for Free-Electron Laser Science, Hamburg, Germany}
\affiliation{The Hamburg Center for Ultrafast Imaging, Hamburg, Germany}

\author{Robert van Leeuwen}
\affiliation{Department of Physics, Nanoscience Center, University of Jyv\"askyl\"a, Jyv\"askyl\"a, Finland}

\author{Klaas J. H. Giesbertz}
\affiliation{Theoretical Chemistry, Faculty of Exact Sciences, VU University, Amsterdam, The Netherlands}
\affiliation{Microsoft Research, AI for Science, Amsterdam, The Netherlands}

\begin{abstract}
We extend a previous result [Sutter
\textit{et al}, \textit{J.\ Phys.\ A: Math.\ Theor.} \textbf{57}, 475202 (2024)] to give an explicit form of the set of $v$-representable densities on the one-dimensional torus with any fixed number of particles in contact with a heat bath at finite temperature. The particle interaction has to satisfy some mild assumptions but is kept entirely general otherwise.
Since we find the necessary domain for densities to be the Sobolev space $H^1$, this leads to a broader dual space for potentials than the usual $L^p$ spaces and encompasses distributions.
By including temperature and thus considering all excited states in the Gibbs ensemble, Gâteaux differentiability of the thermal universal functional is achieved. 
This is equivalent to unique $v$-representability and we can prove that the given set of $v$-representable densities is even maximal.
\end{abstract}
\maketitle

\section{Introduction}\label{sec: intro}
\label{sec:intro}
Density-functional theory (DFT) is based on the one-to-one relation between densities and potentials. The famous work of \citet{hohenberg-kohn1964} introduced an energy functional in terms of the density, which as such is only well-defined for densities coming from a ground state,  a property called `pure-state $v$-representability'. The lack of knowledge about this property, termed the ``$v$-representability problem'', can be circumvented by using the constrained-search functional as introduced by \citet{percus1978} and \citet{levy1979universal}. This functional is defined for all densities coming from an $N$-particle wave function with finite kinetic energy ($N$-representability). However, for the method of \citet{KS1965} it is still crucial to solve the $v$-representability problem. The idea of this method is to reproduce the ground-state density of an interacting Hamiltonian by a non-interacting Hamiltonian with another external potential. Thus, the density needs to be simultaneously $v$-representable for an interacting and a non-interacting system. A possible extension is to allow for mixed states. \citet{levy1982electron} and \citet{Lieb1983} showed that there are ensemble ground-state densities which cannot be realized by a pure ground state.  Hence the set of $v$-representable densities is extended by allowing for ensemble states and not only pure states. In \citet{sutter2024solution}, it is shown that each strictly positive density in $H^1(\T)$ on the one-dimensional torus can be realized by a distributional potential in $H^{-1}(\T)$. Shortly after, \citet{corso2025non, corso2025rigorous} showed that for an odd number of particles on the one-dimensional torus each ground-state density of a Hamiltonian with a potential in $H^{-1}(\T)$ is non-zero. This result completes the full characterization of $v$-representable densities in this case. The same references also show that the external potential is uniquely determined by the ground-state density and that the constrained-search functional is differentiable at strictly positive densities. Additionally, \citet{carvalho2025v} presents these results for non-interacting systems but without any requirement on the number of particles.

Already in 1965 \citet{mermin1965thermal} extended the work of \citet{hohenberg-kohn1964} to non-zero temperatures.
Recently, there has been renewed, broad interest in the field of thermal DFT, ranging from the application of DFT to warm dense matter~\cite{Karasiev2016,Smith2017,Groth2017} and Hubbard models~\cite{Palamara2024} to more theoretical work concerning exact conditions~\cite{Burke2016}, the adiabatic connection~\cite{Harding2024}, a rigorous proof of the Hohenberg--Kohn theorem for systems with temperature~\cite{Garrigue2019}, and the formulation of Entropic DFT~\cite{Yousefi2023}.
For this reason, we want to try and answer the same question about $v$-representability for the case of thermal ensembles and hereby extend our recent work~\cite{sutter2024solution} from zero temperature to elevated temperatures. First results about this problem were achieved by \citet{CCR1985} for infinite lattice systems. \citet{eschrig2010t} formally argued that for the grand canonical ensemble on a $3$-dimensional torus the thermal universal functional is differentiable at $v$-representable densities and that these densities are strictly positive. However, he does not give an explicit expression for the set of $v$-representable densities. In Eschrig's work, the underlying topology is the usually considered Lebesgue space $L^3$, yet the domain of the thermal universal functional (the set of physical densities) is not open in $L^3$. The same problem persists for the usual universal functional of DFT. In particular, arbitrary close to any physical density there is a density that takes negative values. Hence the universal functional is discontinuous at each physical density. This problem was previously studied extensively by \citet{Lammert2007} for the zero temperature case. While Lammert was able to show differentiability at certain densities by considering a Sobolev space topology, he did not characterize the set of all $v$-representable densities, a problem finally solved in \citet{sutter2024solution} on the one-dimensional torus. In our previous and the present work, the problems with openness and continuity are circumvented by considering the finer $H^1$ topology instead of the $L^3$ topology. A topology is finer than another if it has more open sets, or equivalently, less convergent sequences. In one dimension it holds in terms of norms that for any $2\leq p \leq \infty$ there is a $C>0$ such that for all $f\in H^1$ one has
\begin{equation}
    \|f\|_{L^p} \leq C\|f\|_{H^1}
\end{equation}
and one says that $H^1$ is continuously embedded into $L^p$~\cite[Th.~4.12, Part~I.A]{adams-book}. Consequently, in one dimension the $H^1$ topology is finer than all $L^p$ topologies, $p\geq 2$. This leads to the set of strictly positive densities in $H^1$ being open and to a universal functional that is continuous on that set. Hence, we can subsequently investigate Gâteaux differentiability, the property that for a functional $f:X\to\R$ on a Banach space $X$ at every point $x\in U$ from an open domain $U\subseteq X$ the functional derivative
\begin{equation}
    \nabla f(x) = \lim_{\tau\to 0}\frac{f(x+\tau y)-f(x)}{\tau}
\end{equation}
exists for all directions $y\in X$.

In this work, we study the same class of systems as in \citet{sutter2024solution} but at elevated temperatures. This means that we consider a fixed number of fermionic particles on the one-dimensional torus and a very general class of interactions, including of course $W=0$. Elevated temperature means that besides the Hamiltonian an additional entropic term has to be added in the variational principle, and that instead of the usual energy we work with the Helmholtz free energy. The minimizer of the Helmholtz free energy is explicitly known in terms of the Hamiltonian and it is called Gibbs state. As demonstrated before~\cite{SutterGiesbertz2023,GiesbertzRuggenthaler2019}, the entropy term acts as a regularizer and we are not only able to show that each strictly positive density is $v$-representable but we also get (Gâteaux) differentiability of the thermal universal functional for all interactions in the considered class. Additionally, we show that each density coming from a Gibbs state is strictly positive. The same is true for zero temperature and an odd number of particles~\cite{corso2025non, corso2025rigorous}. Hence we get the full characterization of $v$-representable densities. The reason for these additional results is that the minimizer of the Helmholtz free energy does not only occupy the ground state of the Hamiltonian but partially also all excited states. Degenerate states get equal weights and this yields a unique minimizer and strict convexity of the Helmholtz functional (in terms of potentials).  We now formulate the main result of this work.

\vspace{.75em}
\noindent\fbox{\parbox{\textwidth}{\vspace{-.75em}
\begin{theorem}\label{th:main}
    Define a set of densities and a space of (distributional) one-body potentials on the one-dimensional torus $\T$,
    \begin{subequations}
    \begin{align}\label{eq:theorem1-a}
        \vrepset &= \{ \rho\in L^2(\T) \mid \nabla\rho\in L^2(\T), \smallint\rho=N, \forall x\in\T : \rho(x)>0 \}  \\
        \label{eq:theorem1-b}
        \affspace^* &= \{ [v] = \{ v + c \mid c\in \R \} \mid v=f+ \nabla g \;\text{with}\; f,g \in L^2(\T) \}.
    \end{align}
    \end{subequations}
    Then for every $\rho\in\vrepset$ there is a unique potential $v$ with equivalence class $[v]\in\affspace^*$, such that $\rho$ is the density of the Gibbs state (see Proposition~\ref{prop: minimizer of Helmhotz functional/grand potential}) corresponding to the self-adjoint Hamiltonian
    \begin{equation}
        H_v=H_0+V = -\frac{1}{2}\Delta + W + V, \quad V=\sum_{j=1}^N v(x_j),
    \end{equation}
    where $W$ is any (interaction) operator that fulfills the KLMN condition (see Eq.~\eqref{eq:KLMN condition for the interaction}). On the other hand, each density coming from a Gibbs state is contained in $\vrepset$. Moreover, the thermal universal functional $F^\beta_\mathrm{DM}(\rho)$ (defined in Eq.~\eqref{eq:F_DM-const-search}) is Gâteaux differentiable if and only if $\rho\in\vrepset$. This implies that for a fixed interaction $W$, every $\rho \in \vrepset$ is uniquely $v$-representable by a $[v]\in\affspace^*$.
\end{theorem}
\vspace{-.75em}}}
\vspace{.75em}

The property of `unique $v$-representability' is equivalent to the famous Hohenberg--Kohn theorem~\cite{hohenberg-kohn1964,Garrigue2018HK}.
To prove the stated theorem above, we start by introducing the basic setting for an $N$-particle system on a one-dimensional torus and the relevant spaces for densities and potentials (Section~\ref{sec:setting}). Afterwards, we introduce the Helmholtz free energy and present known results about it and its unique minimizer, the Gibbs state, together with the definition of the thermal universal functional (Section~\ref{sec: Helmholtz and Gibbs}). Next, we show the non-emptiness of the subdifferential of the thermal universal functional that is a necessary condition for $v$-representability (Section~\ref{sec:v-rep, subdiff}). This is followed by an application of the preceding results to conclude $v$-representability (Section~\ref{sec: minimum}). We then prove the converse statement, that each Gibbs state indeed yields a density in $\vrepset$, to achieve a full characterization of the set of $v$-representable densities (Section~\ref{sec:circle}). Lastly, we show that the thermal universal functional is Gâteaux differentiable (Section~\ref{sec:Gateaux}). We conclude with the proof of Theorem~\ref{th:main} that is just a collection of previous results and an outlook (Section~\ref{sec:summary}).

\begin{table}[h]
\caption{Notation Table}
\begin{tabular}{ll}
\toprule
Symbol & Description \\
\midrule
$\T$ & one-dimensional torus \\
$L^p$ & real Lebesgue spaces \\
$L^p_{\C}$ & complex Lebesgue spaces \\
$H^1$ & real Sobolev space $W^{1,2}$\\
$H^1_{\C}$ & complex Sobolev space $W^{1,2}$, Eq.~\eqref{eq:complex-H1}\\
$H^{-1}$ & topological dual of the real Sobolev space\\
$\H$ & antisymmetric Hilbert space for $N$ particles, Eq.~\eqref{eq:HS}\\
$\psiset$ & antisymmetric wave function space based on $H^1_\C$, Eq.~\eqref{eq:psiset} \\
$\DMset$ & set of density matrices, Eq.~\eqref{eq:DMset} \\
$\densset$ & set of physical densities, Eq.~\eqref{eq:densset} \\
$\affspace$ & affine space of normalized densities, Eq.~\eqref{eq:affspace} \\
$\vrepset$ & strictly positive subset of the affine space, Eq.~\eqref{eq:theorem1-a} \\
$\affspace^*$ & topological dual of the affine space, Eq.~\eqref{eq:theorem1-b} \\
$\Omega^\beta_v$ & Helmholtz free energy of a density matrices, Eq.~\eqref{eq:Helmholtz-free-energy} \\
$\Omega^\beta$ & Helmholtz functional on potentials, Eq.~\eqref{def: Helmholtz free energy} \\
$F^\beta_\mathrm{DM}$ & thermal universal functional, Eq.~\eqref{eq:F_DM-const-search} \\
\bottomrule
\end{tabular}
\end{table}

\section{One-dimensional torus setting}
\label{sec:setting}

We consider a system at an elevated temperature $T$ consisting of $N$ spin $\frac{1}{2}$-particles. Choosing a different spin as an internal degree of freedom (or none at all) will not change any of the following results, we just use spin $\frac{1}{2}$ and antisymmetric wave functions to relate more closely to the usual systems considered in quantum chemistry. The particles are on the one-dimensional torus $\T$, meaning that they are confined to the interval $[0,1]$ where the points $0$ and $1$ are identified with each other. The underlying complex Hilbert space is then 
\begin{equation}\label{eq:HS}
    \H = \bigl(L^2_{\C}(\T)\otimes\C^2\bigr)^{\wedge N}.
\end{equation}
We also define the complex Sobolev space on the torus
\begin{equation}\label{eq:complex-H1}
    H_{\C}^1(\T)=\{f\in L^2_{\C}(\T)\mid\nabla f\in L^2_{\C}(\T)\}
\end{equation}
(like in Eq.~\eqref{eq:theorem1-a} this means a weak derivative) and the form domain of the kinetic energy operator
\begin{equation}
    Q(-\tfrac{1}{2}\Delta)=\bigl(H_{\C}^1(\T)\otimes \C^2\bigr)^{\wedge N}.
\end{equation}
Together with the normalization constraint this gives the set of wave functions 
\begin{align}\label{eq:psiset}
    \psiset = \{ \Psi \in (H_{\C}^1(\T)\otimes\C^2)^{\wedge N} \mid \|\Psi\|_{\H}=1 \}.
\end{align}
The one-particle density $\rho_\Psi$ for any $\Psi\in\psiset$ is obtained by integrating out all but one of the spatial coordinates from the modulus squared of a wave function summed over all spin components,
\begin{equation}\label{eq:def-dens}
    \rho_\Psi(x) = N \sum_{\sigma_1,\ldots,\sigma_N} \int \abs{\Psi (x\sigma_1, x_2\sigma_2, \dotsc, x_N\sigma_N)}^2 \d x_2 \dotsb \d x_N.
\end{equation}
As already stated in Theorem~\ref{th:main}, the considered Hamiltonians are of the form $H_v=H_0+V$, where $H_0$ consists of the kinetic energy operator and an interaction term.
We generally write $H_v$ instead of $H_{[v]}$, even though the considered potentials come from equivalence classes, since adding a constant to the potential just shifts the eigenvalues while not affecting the eigenstates.
The interaction is kept general and just needs to fulfill the KLMN condition~\cite[Th.~X.17]{reed-simon-2} with respect to the kinetic energy operator. This means that for a fixed interaction $W$ there are $0\leq a< 1$ and $b \geq 0$ such that for each $\Psi\in Q(-\frac{1}{2}\Delta)$ it holds
\begin{align}
    \label{eq:KLMN condition for the interaction}
            |\langle \Psi,W\Psi\rangle | \leq a\langle \Psi,-\tfrac{1}{2}\Delta \Psi \rangle + b \langle \Psi,\Psi \rangle.
\end{align}
This condition guarantees that the Hamiltonian $H_0=-\frac{1}{2}\Delta+W$ is self-adjoint with the same quadratic form domain as the kinetic energy operator. In \citet[Appendix]{sutter2024solution} this condition was shown to hold for interactions in $L^1([-1,1])$ and for distributional interactions of the form $\nabla g$ with $g \in L^2([-1,1])$. Here and in the following we use the notion of quadratic forms instead of linear operators~\cite[Sec.~VIII.6]{reed-simon-1}. For the kinetic energy operator this is achieved directly by integration by parts
\begin{equation}\label{eq:Laplace-expectation-val}
    \frac{1}{2}\langle \Psi, -\Delta \Psi \rangle =\frac{1}{2} \sum_{j=1}^N \langle \Psi, -\Delta_j \Psi \rangle = \frac{1}{2}\sum_{j=1}^N \langle \nabla_j \Psi, \nabla_j \Psi \rangle = T(\Psi).
\end{equation}
In the following theorem we collect some important previous results~\cite{Lieb1983,sutter2024solution}. 
\begin{theorem}\label{th:N-rep and properties of densset}
    We define the following set of physical densities,
    \begin{equation}\label{eq:densset}
    \densset = \{ \rho \in L^1(\T) \mid \nabla\sqrt{\rho} \in L^2(\T), \rho \geq 0, \smallint \rho = N \},
\end{equation}
and an affine space of normalized densities in the real Sobolev space $H^1(\T)$, 
\begin{equation}\label{eq:affspace}
    \affspace = \{ \rho \in H^1(\T) \mid \smallint \rho = N \}.
\end{equation}
Then, the following holds.
\begin{enumerate}[label=(\roman*)]
    \item 	Let $\Psi \in \psiset$. Then $\rho_\Psi\in\densset$ and it holds
    \begin{equation}\label{eq:psi-dens-estimate-1}
        \int \left|\nabla\sqrt{\rho_\Psi(x)}\right|^2 \d x \leq 2T(\Psi).
    \end{equation}
    \item     There exist two constants $C_1,C_2 > 0$ such that for every $\rho \in \densset$ there is a $\Psi \in \psiset$ that has $\rho_\Psi=\rho$ and fulfills the following estimate,
    \begin{equation}\label{eq:psi-dens-estimate-2}
        T(\Psi)\leq C_1 + C_2 \int \left| \nabla \sqrt{\rho(x)}\right | ^2 \d x.
    \end{equation}
    \item\label{item:embeddings} $\vrepset \subset \densset \subset \affspace \subset H^{1}(\T)\hookrightarrow C^0(\T) \hookrightarrow L^\infty(\T) \hookrightarrow L^3(\T) \hookrightarrow L^2(\T) \hookrightarrow L^1(\T)$, where $X \hookrightarrow Y$ means that $X$ is continuously embedded in $Y$.
\end{enumerate}
\end{theorem}

\citet{Lieb1983} showed that the set $\densset$ is convex. This implies that (finite) convex combinations of elements in $\densset$ lie again in $\densset$. Motivated by the observations for the zero temperature case~\cite{sutter2024solution}, we identify the space of external potential by $\affspace^*$, the topological dual of the subspace $\{\rho\in H^1(\T)\mid\int\rho=0\}$ that is isomorphic to $\affspace$. This space $\affspace^*$ contains equivalence classes, where the elements are from the space $H^{-1}(\T)$ and where two elements are in the same class if and only if they differ only by a constant, i.e., like in Eq.~\eqref{eq:theorem1-b},
\begin{equation}\label{eq:v-dual-space}
    \affspace^* = \{ [v] = \{ v + c \mid c\in \R \} \mid v=f+ \nabla g \;\text{with}\; f,g \in L^2(\T) \}.
\end{equation}
The dual pairing of an element $v=f+\nabla g \in H^{-1}(\T)$ with an element $\varphi \in H^1(\T)$ is given by $v(\varphi) = \langle v,\varphi \rangle = \langle f,\varphi \rangle - \langle g,\nabla \varphi \rangle$. Using $\affspace^*$ as potential space automatically guarantees that the potential energy is finite. That Hamiltonians of the form $H=H_0+V$ with $V$ defined by a $v \in \affspace^*$ are connected to a self-adjoint operator was shown in \citet[Th.~18]{sutter2024solution} by using the KLMN theorem~\cite[Th.~X.17]{reed-simon-2}.

\section{Helmholtz free energy and Gibbs state}\label{sec: Helmholtz and Gibbs}

Considering a system at fixed temperature $T>0$ (through coupling to a heat bath) means that we cannot rely on a variational principle for the internal energy $E_v(\Gamma) = \Tr\{\Gamma H_v\}$ any more \cite{binder2018thermodynamics}. Instead, we have to consider the Helmholtz free energy for arbitrary density matrices, 
\begin{equation}\label{eq:Helmholtz-free-energy}
    \Omega_v^\beta(\Gamma) = E_v(\Gamma) - T S(\Gamma) = \Tr \bigl\{\Gamma\bigl(H_v+\beta^{-1}\log(\Gamma)\bigr)\bigr\},
\end{equation}
where $S(\Gamma)=-\Tr \{\Gamma \log(\Gamma)\}$ is the entropy and $\beta=T^{-1}$ is the inverse temperature.
The variational principle then yields the Helmholtz functional,
\begin{equation}\label{def: Helmholtz free energy}
        \Omega^\beta(v)=\inf_{\Gamma\in\DMset}\Omega^\beta_v(\Gamma)= \inf_{\Gamma\in\DMset}\Tr \bigl\{\Gamma\bigl(H_v+\beta^{-1}\log(\Gamma)\bigr)\bigr\}.
\end{equation}
Here the infimum is taken over the set of all density matrices,
\begin{equation}\label{eq:DMset}
    \DMset = \left\{ \Gamma=\sum_{k=1}^\infty \lambda_k \langle \Psi_k,\cdot \rangle \Psi_k \;\middle|\; \lambda_k\geq 0, \Tr\{\Gamma\}=1, \Psi_k \in \psiset \;\text{all orthogonal}, S(\Gamma) <\infty \right\}.
\end{equation}
Making the Helmholtz free energy stationary with respect to small perturbations $\delta \Gamma$ with $\Tr\{\delta \Gamma \}=0$ gives the minimizer. 
\begin{proposition}
\label{prop: minimizer of Helmhotz functional/grand potential}
    The unique minimizer of the Helmholtz free energy is given by $\Gamma_v=\e^{-\beta H_v}/Z(v) \in\DMset$, where $Z(v)=\Tr \{\e^{-\beta H_v}\}$ is called the partition function. The state $\Gamma_v$ is called the Gibbs state and only exists if $Z(v)$ is finite. Moreover, for each $v\in H^{-1}(\T)$ the Gibbs state exists and the Helmholtz free energy is given by 
    \begin{align}\label{eq: Helmholtz free energy}
        \Omega_v^\beta(\Gamma_v)= \Omega^\beta(v)=-\beta^{-1}\log\bigl(Z(v)\bigr).
    \end{align}
\end{proposition}

\begin{remark}
For an example of a physical system with infinite partition function simply consider the hydrogen atom in the infinite-volume space $\R^3$~\cite[Sec.~2.3]{GiesbertzRuggenthaler2019}. The eigenenergies are $-1/(2n^2)$ with multiplicity $n^2$, $n \in \N$. Then the partition function is 
\begin{equation}
    Z(v)=\sum \limits_{n=1}^\infty n^2 \e^{\beta/(2n^2)}\geq \sum \limits_{n=1}^\infty n^2=\infty.
\end{equation}
\end{remark}

The proof of Proposition~\ref{prop: minimizer of Helmhotz functional/grand potential} can be found in Appendix \ref{app: proof minimizer properties}.  
Proposition \ref{prop: minimizer of Helmhotz functional/grand potential} implies that the density of the Gibbs state is an \emph{infinite} convex combination of elements in $\densset$. Hence we cannot directly conclude that its density is contained in $\densset$, but we will show this statement with the following proposition.
\begin{proposition}\label{prop: Gibbs state density in densset}
    The density of the Gibbs state $\Gamma_v$ is contained in $\densset$ for all $v \in H^{-1}(\T)$.
\end{proposition}
The proposition is proven in Appendix \ref{app: proof Gibbs state density in densset}. It tells us that $v$-representable densities at elevated temperatures are always contained in $\densset$. Hence we can split the minimization of the Helmholtz free energy $\Omega^\beta_v$ into first a variation over densities in $\densset$ followed by a minimization over all density matrices which yield such a density, just like in the zero temperature case,
\begin{subequations}
\begin{align}\label{eq: Omega-from-F}
    &\Omega^\beta(v)=\inf_{\rho \in \densset}\{F^\beta_\mathrm{DM}(\rho)+\innerproduct{v,\rho}\} \quad \text{where} \\
    \label{eq:F_DM-const-search}
    &F^\beta_\mathrm{DM}(\rho)=\inf_{\substack{\Gamma\in\DMset \\ \Gamma \mapsto \rho}} \Tr \{\Gamma(H_0+\beta^{-1}\log(\Gamma))\}.
\end{align}
\end{subequations}
The functional $F_{\mathrm{DM}}^\beta(\rho)$ is called the thermal universal functional and we can extend it to $\affspace$ by setting $F_{\mathrm{DM}}^\beta(\rho)=\infty$ for all $\rho \in \affspace \setminus \densset$. Note that because of the existence of the Gibbs state and Propositions \ref{prop: minimizer of Helmhotz functional/grand potential} and \ref{prop: Gibbs state density in densset}, the infimum in the variation over the densities is always attained, whereas for the second infimum (in the definition of $F_\mathrm{DM}^\beta(\rho)$) we only know about the existence of a minimizer for $v$-representable densities which are characterized in Theorem \ref{th:main}.

Like in the zero-temperature case~\cite{Lieb1983,sutter2024solution}, we would like to have a convex formulation of thermal density-functional theory to exploit the properties of convex functions. Due to the concavity of the entropy~\cite[Th.~12]{GiesbertzRuggenthaler2019} and linearity of $\Gamma \to \rho$ the following proposition follows.
\begin{proposition}\label{prop:convexity of Helmholtz and universal functional}
The following two statements hold.
    \begin{enumerate}[label=(\roman*)]
    \item The Helmholtz free energy $\Gamma\mapsto\Omega^\beta_v(\Gamma)$ is convex on $\DMset$.
    \item The thermal universal functional $\rho\mapsto F^\beta_\mathrm{DM}(\rho)$ is convex on $\affspace$.
\end{enumerate}
\end{proposition}

\section{Relation between $v$-representability and non-empty subdifferential}\label{sec:v-rep, subdiff}
The relation between $v$-representability and non-empty subdifferential is usually discussed in mathematical presentations of DFT~\cite{eschrig2003-book,Penz-et-al-HKreview-partI,cances2023density}. This relation, together with some additional results from convex analysis, are important ingredients to show our main Theorem \ref{th:main}. Therefore, we will repeat the relation and the needed results in this section. 

For any proper ($f \neq \infty$) convex function $f: \affspace \to (-\infty, \infty]$ its subdifferential at a point $\rho \in \affspace$ is defined as the set 
\begin{equation}
    \partial f(\rho) = \{ u\in \affspace^* \mid  f(\rho)-f(\rho') \leq \langle u,\rho-\rho' \rangle \; \forall \rho'\in\affspace\}.
\end{equation}
Note, that if the inequality $f(\rho)-f(\rho') \leq \langle u,\rho-\rho'\rangle$ cannot be satisfied for any $u \in \affspace^*$ and all $\rho' \in \affspace$ then the subdifferential is empty. An element of $\partial f(\rho)$ is called a subgradient. Moreover, the function $f$ is Gâteaux differentiable at a point $\rho$ if the function $f$ is continuous at $\rho$ and if the subdifferential $\partial f(\rho)$ contains only a single element~\cite[Prop.~2.40]{Barbu-Precupanu}. In this case, this element is equal to the derivative. If the general convex function is replaced by the thermal universal functional (which is a proper convex functional by Proposition \ref{prop:convexity of Helmholtz and universal functional}), we get for an element $v \in -\partial F_\mathrm{DM}^\beta(\rho)$  and for all $\rho' \in \affspace$ the following inequality 
\begin{align}\label{eq:FDM-inequ}
    F_\mathrm{DM}^\beta(\rho)+\langle v, \rho \rangle \leq F_\mathrm{DM}^\beta(\rho')+\langle v,\rho'\rangle.
\end{align}
In particular, $\rho$ is a minimizer for the variation in Eq.~\eqref{eq: Omega-from-F}. Note that this is only a necessary condition for $v$-representability and not a sufficient one. The reason is that for $v$-representability of a density $\rho$ we additionally require that the variation over $\Gamma \mapsto \rho$ in Eq.~\eqref{eq:F_DM-const-search} has a minimum too, otherwise there is no Gibbs state behind the density.
In this section, we show that the subdifferential of $F^\beta_\mathrm{DM}(\rho)$ at any density $\rho \in \vrepset$ is non-empty. In the next section we then show that these densities indeed come from Gibbs states. Only the two result together then give $v$-representability for all densities in $\vrepset$. 

To show the existence of a subgradient for a density $\rho \in \vrepset$ we proceed as in \citet{sutter2024solution}. 
The only result that significantly changes compared to our previous work is the local boundedness of the thermal universal functional $F_\mathrm{DM}^\beta$.  

\begin{lemma}\label{lem:F-bound}
    There exist three constants $C_1, C_2, C_3 > 0$ such that for all $\rho \in \densset$ it holds
    \begin{equation}
        C_1 \leq F_\mathrm{DM}^\beta(\rho) \leq C_2 + C_3 \| \nabla \sqrt{\rho} \|_{L^2}^2.    
    \end{equation}
\end{lemma}
\begin{proof}
    For the lower bound of $F_\mathrm{DM}^\beta(\rho)$ we have
    \begin{align}
        F_\mathrm{DM}^\beta(\rho)\geq \inf_{\Gamma\in\DMset} \Tr \{\Gamma(H_0+\beta^{-1}\log(\Gamma))\}=\Omega^\beta(0),
    \end{align}
    which is greater than $-\infty$ by Proposition \ref{prop: minimizer of Helmhotz functional/grand potential}. For the upper bound we note that the entropic term $-\beta^{-1} S(\Gamma)$ is negative. Hence we get the following inequality
    \begin{align}
        F_\mathrm{DM}^\beta(\rho)\leq \inf_{\substack{\Gamma\in\DMset \\ \Gamma \mapsto \rho}} \Tr \{\Gamma H_0\}\leq C_2 + C_3 \| \nabla \sqrt{\rho} \|_{L^2}^2,
    \end{align}
    where we used a previous result~\cite[Lem.~10]{sutter2024solution} for the second inequality.  
\end{proof}
\begin{theorem}\label{th:F-nonempty-subdiff}
    At any $\rho \in \vrepset$ the thermal universal functional $F_\mathrm{DM}^\beta$ is continuous and the subdifferential $\partial F_\mathrm{DM}^\beta(\rho) \subset \affspace^*$ is non-empty.
\end{theorem}
The proof is the same as in \citet[Th.~15]{sutter2024solution}.
The result relies critically on the embedding $H^1\hookrightarrow L^\infty$ from Theorem~\ref{th:N-rep and properties of densset}\ref{item:embeddings} that can be used to show that $\vrepset$ is an open set~\cite[Lem.~14]{sutter2024solution}. This embedding only holds true in one dimension and one further needs the torus setting with a non-vanishing minimum of $\rho$ in order to show that $F_\mathrm{DM}^\beta$ is bounded above. Both results combined then yield the desired result by using powerful theorems from convex analysis~\cite[Th.~2.14 and Prop.~2.36]{Barbu-Precupanu}. We note that the same result can be established also on the open interval with Neumann boundary conditions~\cite{corso2025rigorous}. We can now also understand the origin of the distributional potentials. We have $\vrepset \subset \affspace \subset H^1(\T)$, so the density domain is limited to a subset of the Sobolev space $H^1(\T)$ with dual space $H^{-1}(\T)\subset\affspace^*$ that contains distributions. 
On the other hand, a potential is always from the subdifferential of the thermal universal functional and the subdifferential itself is a subset of the dual space $\affspace^*$.

\section{Existence of the Minimum}\label{sec: minimum}

As discussed in the previous section, we do not only need a non-empty subdifferential of $F_\mathrm{DM}^\beta(\rho)$ to get $v$-representability but also that the minimum in the definition of $F_\mathrm{DM}^\beta(\rho)$ (Eq.~\eqref{eq:F_DM-const-search}) exists. This minimizer is then the Gibbs state $\Gamma_v$, where $v$ is just the potential that yields $\rho$. In the following we show this for all $\rho\in \affspace_{>0}$. To do so, we make use of the relative entropy which for two non-negative trace class operators $A$, $B$ with $\ker B \subset \ker A$ is defined as
\begin{align}\label{def: relative entropy}
    S(A | B) = \Tr\{A(\log A- \log B)+B-A\}.
\end{align}
The relative entropy is always non-negative, $S(A |B) \geq 0$ and $S(A|B)=0$ if and only if $A=B$~\cite{Lindblad1973,Falk1970}. An important property of the relative entropy is its lower semi-continuity. 
\begin{lemma}[\citet{lami2023attainability}, Lemma 4]
\label{lemma: lower semi-continuity of the relative entropy}
    The relative entropy $(A , B) \mapsto S(A | B)$ is jointly lower semi-continuous with respect to the weak-* topology.
\end{lemma}
The next lemma relates the Helmholtz free energy and the relative entropy. Note that originally we only defined the Helmholtz free energy $\Omega^\beta_v(\Gamma)$ for $\Gamma\in\DMset$. However, in the same manner, it can be defined for any non-negative trace-class operator with finite entropy. Recall that $\Gamma_v$ is the Gibbs state corresponding to a potential $v$ and is given by $\Gamma_v=\e^{-\beta H_v}/Z(v)$. Its kernel is trivial and thus $S(\cdot | \Gamma_v)$ is well-defined.
\begin{lemma}[\citet{GiesbertzRuggenthaler2019}]
\label{lemma: relation relative entropy and Helmholtz functional}
    $S(\Gamma | \Gamma_v)+\Tr\{\Gamma- \Gamma_v\}=\beta \Omega^\beta_v(\Gamma)- \beta \Omega^\beta_v(\Gamma_v)\Tr\{\Gamma\}$.
\end{lemma}
Note that the left-hand side is well-defined for all non-negative trace-class operators $\Gamma$ and thus Lemma \ref{lemma: relation relative entropy and Helmholtz functional} extends the domain of $\Omega_v^\beta(\cdot)$ to density matrices with infinite entropy.
Since the left-hand side in Lemma \ref{lemma: relation relative entropy and Helmholtz functional} consists of two weak-* lower semi-continuous  functions (the trace norm $\Tr|\Gamma|$ is decreasing under weak-* limits and equal to $\Tr\{\Gamma\}$ since $\Gamma\geq 0$), we get that the Helmholtz free energy together with a term proportional to $\Tr\{\cdot\}$ (with positive factor) is lower semi-continuous with respect to the weak-* topology. In particular, we get the following lemma.
\begin{lemma}
\label{lemma: assumtion on the energy}
    For $\xi > 1$ the mapping $\Gamma\mapsto\Omega^\beta_v(\Gamma)-\xi \Omega^\beta_v(\Gamma_v)\Tr\{\Gamma\}$ is weak-* lower semi-continuous if $\Omega^\beta_v(\Gamma_v)\leq(\beta(\xi-1))^{-1}$. For $\xi<1$ the function is weak-* lower semi-continuous if $\Omega^\beta_v(\Gamma_v)\geq (\beta(\xi-1))^{-1}$.
\end{lemma}
\begin{proof}
    From Lemma \ref{lemma: relation relative entropy and Helmholtz functional} we get
    \begin{equation}
        \Omega^\beta_v(\Gamma)-\xi \Omega^\beta_v(\Gamma_v) \Tr\{\Gamma\}=\frac{1}{\beta}S(\Gamma | \Gamma_v)-\frac{1}{\beta}\Tr\{\Gamma_v\} + \Tr\{\Gamma\}\left(\frac{1}{\beta}+(1-\xi)\Omega^\beta_v(\Gamma_v) \right).
    \end{equation}
    By Lemma \ref{lemma: lower semi-continuity of the relative entropy}  we know that the first term $\beta^{-1}S(\Gamma | \Gamma_v)$ is weak-* lower semi-continuous. The second term $-\beta^{-1}\Tr\{\Gamma_v\}$ is just a constant. The last term is weak-* lower semi-continuous if $\beta^{-1}+(1-\xi )\Omega^\beta_v(\Gamma_v)\geq 0$. Solving for $\Omega^\beta_v(\Gamma_v)$ proves the lemma.
\end{proof} 

\begin{theorem}\label{th: existence of the minimizer for the universal functional}
    For every $\rho \in \vrepset$ and all potentials $v\in -\partial F_\mathrm{DM}^\beta(\rho)\subset\affspace^*$, it holds that the corresponding Gibbs state $\Gamma_v$ yields the density $\rho$ and is the minimizer in Eq.~\eqref{eq:F_DM-const-search}, i.e., $F_\mathrm{DM}^\beta(\rho)= \Tr\{\Gamma_v(H_0+\beta^{-1}\log \Gamma_v)\}$.
\end{theorem}
\begin{proof}
For a density $\rho \in \vrepset$ we pick $v \in -\partial F_\mathrm{DM}^\beta(\rho)$ which we know to exist by Theorem \ref{th:F-nonempty-subdiff}.  By definition of $F_\mathrm{DM}^\beta(\rho)$ there is a sequence of density matrices $\Gamma_n$ each yielding the density $\rho$ and such that $\Tr \{\Gamma_n(H_0+\beta^{-1}\log \Gamma_n)\} \to F^\beta_\mathrm{DM}(\rho)$. Since all $\Gamma_n$ are contained in a ball of radius $1$ we get by the Banach--Alaoglu theorem \cite[Th.~IV.21]{reed-simon-1} that there is a subsequence (in the following also denoted by $\Gamma_n$) and a trace class operator $\Gamma_0$ such that $\Gamma_n \overset{\ast}{\rightharpoonup} \Gamma_0$. Moreover, we know that $\Gamma_0$ is non-negative and that $\Tr \{\Gamma_0\} \leq 1$. To prove the theorem, we show that $\Gamma_0$ is a minimizer in Eq.~\eqref{eq:F_DM-const-search} that, together with  the discussion after Eq.~\eqref{eq:FDM-inequ}, already shows that $\Gamma_0$ is equal to the Gibbs state $\Gamma_v\in\DMset$. This means that $\Gamma_0$ is normalized to $1$, that $F_\mathrm{DM}^\beta(\rho)=\Tr \{\Gamma_0(H_0+\beta^{-1}\log \Gamma_0)\}$, and that it yields the density $\rho$. We start by showing $\Tr \{\Gamma_0\}=1$.\\
For simplicity, we call the minimum of the Helmholtz free energy $\Omega^\beta=\Omega^\beta_v(\Gamma_v)=\inf_{\Gamma \in \DMset} \Tr \{\Gamma(H_v+\beta^{-1}\log \Gamma)\}$. Because $v \in -\partial F^\beta_\mathrm{DM}(\rho)$ and $\lim_{n \to \infty} \Tr\{\Gamma_n\}=1$ we have 
\begin{equation}
\label{achievment of infimum: lower semi continuity for limit for the canonical ensemble}
\begin{aligned}
    \Omega^\beta-\xi \Omega^\beta&=F^\beta_\mathrm{DM}(\rho)+\innerproduct{v, \rho}-\xi \Omega^\beta=\lim_{n \to \infty}\Tr \{\Gamma_n(H_v+\beta^{-1}\log \Gamma_n)\}-\xi \Omega^\beta \Tr\{\Gamma_n\} \\
    &\geq \Tr \{\Gamma_0(H_v+\beta^{-1}\log \Gamma_0)\}-\xi \Omega^\beta \Tr\{\Gamma_0\},
\end{aligned}
\end{equation}
where we have used the fact that the function $\Gamma\mapsto\Omega^\beta_v(\Gamma)-\xi \Omega^\beta \Tr\{\Gamma\}$ is weak-* lower semi-continuous due to Lemma \ref{lemma: assumtion on the energy}. For the condition of the lemma to be fulfilled we set $\xi=(\beta\Omega^\beta)^{-1}+1$. This expression is well-defined, since we can choose a representative of $[v] \in \affspace^*$ for which $\Omega^\beta_v(\Gamma_v)$ is non-zero. Let $c$ be the constant such that $\Tr \{c \Gamma_0\}=1$. (Note that the constant $c$ only exist if $\Gamma_0 \neq 0$. If $\Gamma_0=0$ then the right-hand side of Eq.~\eqref{achievment of infimum: lower semi continuity for limit  for the canonical ensemble} is $0$ while the left-hand side is $-\beta^{-1}$. Hence $\Gamma_0=0$ is not possible.) Since $\Tr \{\Gamma_0\} \leq 1$ it holds $c \geq 1$. We compute
\begin{equation}
\begin{aligned}
    \Omega^\beta-\xi \Omega^\beta \leq \Omega^\beta_v(c\Gamma_0)-\xi \Omega^\beta \Tr\{c \Gamma_0\}&=\Tr \{c\Gamma_0(H_v+\beta^{-1}\log(c\Gamma_0))\}-\xi \Omega^\beta \Tr\{c \Gamma_0\} \\
    &=c\Tr \{\Gamma_0(H_v+\beta^{-1}\log \Gamma_0)\}+\beta^{-1}\Tr \{c\Gamma_0\} \log c-c\xi \Omega^\beta \Tr\{\Gamma_0\}  \\
    &=c(\Tr \{\Gamma_0(H_v+\beta^{-1}\log \Gamma_0)\}-\xi \Omega^\beta \Tr\{\Gamma_0\})+\beta^{-1} \log c  \\
    &\leq c (\Omega^\beta-\xi \Omega^\beta)+\beta^{-1}\log c,
\end{aligned}
\end{equation}
where we used the variational principle for $\Omega^\beta\leq\Omega^\beta_v(c\Gamma_0)$, the fact that $\Tr \{c \Gamma_0\}=1$, and Eq.~\eqref{achievment of infimum: lower semi continuity for limit for the canonical ensemble}. This inequality can be transformed into
\begin{align}
    \label{inequ: achievement of infimum}
    c\Omega^\beta-c\xi \Omega^\beta+\beta^{-1}\log c -\Omega^\beta+\xi \Omega^\beta \geq 0.
\end{align}
 It is clear that we have equality if $c=1$. 
The function $f(c)=c\Omega^\beta-c\xi \Omega^\beta+\beta^{-1}\log c -\Omega^\beta+\xi \Omega^\beta$ defined on $[1, \infty)$ has $f(1)=0$ and the derivative
\begin{equation}
f'(c)=\Omega^\beta - \xi \Omega^\beta + \frac{\beta^{-1}}{c}=-\beta^{-1}+\frac{\beta^{-1}}{c}\begin{cases}
    =0 \quad& c=1\\
    <0 \quad& c>1,
\end{cases}
\end{equation}
where we used $\xi=(\beta\Omega^\beta)^{-1}+1$. This means that $f(c)$ is strictly decreasing on $[1,\infty)$ with a unique maximum at $c=1$. Hence the inequality \eqref{inequ: achievement of infimum} can only be satisfied if $c=1$ and we conclude that $\Tr \{\Gamma_0\}=1$.\\
Now that we know that the weak-* limit $\Gamma_0$ has trace norm equal to $1$ we can show that $\Gamma_0 \mapsto \rho$. For this we use the same idea as \citet{Lieb1983} and show that for any $f \in L^\infty(\T)$ we have $\Tr \{\Gamma_n M_f\} \to \Tr \{\Gamma_0 M_f\}$, where $M_f$ denotes pointwise multiplication by $f$. This means that the densities of $\Gamma_n$ converge weakly to the density of $\Gamma_0$ (in $L^1$) and by the uniqueness of weak limits we get $\Gamma_0 \mapsto \rho$. To do this, let $\sum _j \lambda_j \ket{\psi_j} \bra{\psi_j}$ be the spectral decomposition of $\Gamma_0$. Since $\Tr \{\Gamma_0\}=1$ there is for any $\epsilon>0$ an $N_\epsilon \in \N$ such that 
\begin{equation}
    \sum \limits_{j=1}^{N_\epsilon} \lambda_j > 1- \epsilon.
\end{equation}
We define the projection 
\begin{equation}
    P=\sum \limits_{j=1}^{N_\epsilon} \ket{\psi_j}\bra{\psi_j}.
\end{equation}
The operator $P$ is compact and since $\Gamma_n \overset{\ast}{\rightharpoonup} \Gamma_0$ is trace class there exists an $N_P$ such that $\abs{\Tr \{\Gamma_n P\}-\Tr \{\Gamma_0 P\}}< \epsilon$ for all $n>N_P$. It follows that for $n>N_P$ we have $\Tr \{\Gamma_nP\}>\Tr \{\Gamma_0 P\}-\epsilon > 1- 2 \epsilon$ and thus $\Tr \{\Gamma_n(1-P)\}<2 \epsilon$. We thus get that $\abs{\Tr \{(\Gamma_n-\Gamma_0)(1-P)\}}<2 \epsilon$ for large enough $n$ and we can write
\begin{equation}
    \Tr \{(\Gamma_n-\Gamma_0)M_f\}=\Tr \{(\Gamma_n-\Gamma_0)PM_f\}+\Tr \{(\Gamma_n-\Gamma_0)(1-P)M_f\}.
\end{equation}
The first term goes to $0$ since $PM_f$ is a compact operator. The second term is contained in $[-C\epsilon, C\epsilon]$ for $n$ large enough where the constant $C$ depends only on $f$. Since $\epsilon$ can be arbitrarily small we conclude that $\Tr \{(\Gamma_n-\Gamma_0)M_f\} \to 0$ and hence $\Gamma_0 \mapsto \rho$. By Eq.~\eqref{achievment of infimum: lower semi continuity for limit for the canonical ensemble} we get $F_\mathrm{DM}^\beta(\rho)=\Tr \{\Gamma_0(H_0+\beta^{-1}\log \Gamma_0)\}$ which concludes the proof. 
\end{proof}

\begin{remark}
\label{rmk: v-rep and subdiff equivalence}
    Note that we never use $\rho>0$ and that we can also replace the condition $\rho \in \vrepset$ by the condition that the subdifferential of $F_\mathrm{DM}^\beta$ at $\rho$ is non-empty.
\end{remark}

If we combine Theorem~\ref{th:F-nonempty-subdiff} and Theorem~\ref{th: existence of the minimizer for the universal functional}, we get that $\rho \in \vrepset$ is a sufficient condition for $v$-representability by a Gibbs state, which is the first part of Theorem~\ref{th:main}. The full proof for Theorem~\ref{th:main} can be found in Section~\ref{sec:summary}.

\section{Closing the Circle}\label{sec:circle}
So far we showed that each density in $\vrepset$ is $v$-representable by a Gibbs state. In this section we show that each $v$-representable density is strictly positive and thus contained in $\vrepset$. Consequently, the restriction to strictly positive densities in order to show $v$-representability of a density is necessary. In order to show this, observe that at elevated temperatures, all excited states of the Hamiltonian contribute to the Gibbs state (and hence also to the density) and since these states form an orthonormal basis, it is not possible that the density vanishes somewhere. We formulate this result in the following theorem that closes the circle from strictly positive densities in $\affspace_{>0}$ to representing potentials in $\affspace^*$ and back to a strictly positive density of the corresponding Gibbs state.
\begin{theorem}
\label{th: Gibbs state density non-zero}
    The density of any Gibbs state is strictly positive.
\end{theorem}

\begin{proof}
Let $v \in H^{-1}(\T)$. The Gibbs state is given by $\Gamma_v=\e^{-\beta H_v}/Z(v)$ which can be written as 
\begin{equation}
    \Gamma_v=\frac{1}{Z(v)}\sum_j \e^{-\beta \lambda_j}\ket{\psi_j}\bra{\psi_j},
\end{equation}
where the $\lambda_j$ are the eigenvalues of $H_v$ with corresponding eigenstates $\psi_j\in\psiset$. We denote the density of $\Gamma_v$ by $\rho$ and the densities of $\psi_j$ by $\rho_j$. Let us assume that the density $\rho$ is equal to $0$ at $x_0$. By the relation 
\begin{equation}
    \rho(x)=\frac{1}{Z(v)}\sum_j \e^{-\beta \lambda_j} \rho_j(x)   
\end{equation}
we conclude that $\rho_j(x_0)=0$ for all $j$ since the weights are all non-zero. This implies that all $\psi_j$ vanish on $U=\{x_0\} \times [0,1]^{N-1} \times \{\uparrow, \downarrow\}^N$. Note that the restriction to $U$ of $\psi_j\in\psiset$ or any linear combination of those vectors in $(H_{\C}^1(\T)\otimes\C^2)^{\wedge N}$ is well-defined by the trace operator for Sobolev spaces~\cite[p.~315]{brezis2011functional}, which is not the case for a general element in $\H=(L^2_{\C}(\T)\otimes\C^2)^{\wedge N}$. Since the $\psi_j$ form a basis of $(L^2_\C(\T)\otimes\C^2)^{\wedge N}$ we also need to be able to express each function in $(H_{\C}^1(\T)\otimes\C^2)^{\wedge N}\subset (L^2_\C(\T)\otimes\C^2)^{\wedge N}$ as their linear combination. Now this yields a contradiction since clearly not every function in $(H_{\C}^1(\T)\otimes\C^2)^{\wedge N}$ vanishes on $U$.
\end{proof}

\section{Gâteaux differentiability}\label{sec:Gateaux}
The last part of the main result (Theorem \ref{th:main}) is about Gâteaux differentiability of $F_\mathrm{DM}^\beta$ at densities $\rho \in \vrepset$. As mentioned in Section~\ref{sec:v-rep, subdiff}, Gâteaux differentiability can be achieved by continuity and a single-valued subdifferential. Thus, the goal of this section is to show that each density $\rho \in \vrepset$ has only one subgradient. We do that by using a standard result from convex analysis and by the usual Hohenberg--Kohn argument.

\setcounter{theorem}{12}

\begin{lemma}[\citet{Barbu-Precupanu}, Prop.~2.40]\label{lem: cond. for Gateaux differentiability}
    If a convex function $f$ on $\affspace$ is Gâteaux differentiable at $\rho$, then $\partial f(\rho)$ consists of a single element $v=\mathrm{grad}f(\rho)$. Conversely, if $f$ is continuous at $\rho$ and if $\partial f(\rho)$ contains only a single element, then $f$ is Gâteaux differentiable at $\rho$ and $\mathrm{grad} f(\rho)=\partial f(\rho)$.
\end{lemma}
We now have all the necessary ingredients to show Gâteaux differentiability of the thermal universal functional.
\begin{theorem}\label{th: differentiability of universal functional}
    The thermal universal functional $F_\mathrm{DM}^\beta$ is Gâteaux differentiable at any density $\rho \in \vrepset$.
\end{theorem}

\begin{proof}
According to Theorem \ref{th:F-nonempty-subdiff}, the functional $F_\mathrm{DM}^\beta$ is continuous on $\vrepset$ and $\partial F_\mathrm{DM}^\beta (\rho)$ is non-empty for all densities $\rho \in \vrepset$. By Lemma \ref{lem: cond. for Gateaux differentiability} it is then enough to show that $\partial F_\mathrm{DM}^\beta (\rho)$ is single-valued. Let us assume that there is a $\rho \in \vrepset$ for which $\partial F_\mathrm{DM}^\beta(\rho)$ is not single-valued. In particular, we assume that $v_1,v_2 \in -\partial F_\mathrm{DM}^\beta(\rho)$ with $v_1 \neq v_2$. By Theorem \ref{th: existence of the minimizer for the universal functional} the two Gibbs states $\Gamma_1$ and $\Gamma_2$ corresponding to $v_1$ and $v_2$ both yield the density $\rho$. Moreover, Proposition \ref{prop: minimizer of Helmhotz functional/grand potential} states $\Gamma_1 \neq \Gamma_2$. We then argue like in the proof of the Hohenberg--Kohn theorem~\cite{hohenberg-kohn1964} that
\begin{equation}\begin{aligned}
     \Omega^\beta(v_1)&=\Omega^\beta_{v_1}(\Gamma_1)=\Omega^\beta_{v_2}(\Gamma_1)+\innerproduct{\rho,v_1-v_2} > \Omega^\beta_{v_2}(\Gamma_2)+\innerproduct{\rho,v_1-v_2}= \Omega^\beta(v_2)+\innerproduct{\rho,v_1-v_2}.
\end{aligned}\end{equation}
Changing the roles of $v_1$ and $v_2$ and adding the two inequalities gives the contradiction
\begin{equation}
    \Omega^\beta(v_1)+\Omega^\beta(v_2)<\Omega^\beta(v_2)+\Omega^\beta(v_1).
\end{equation}
Hence, $\partial F_\mathrm{DM}^\beta (\rho)$ must be single-valued.
\end{proof}

\section{Summary and outlook}
\label{sec:summary} 

We can now collect the results from the previous sections to finally give a proof of our main result.

\begin{proof}[Proof of Theorem \ref{th:main}]
Let $\rho \in \vrepset$, then Theorem \ref{th:F-nonempty-subdiff} shows that $\partial F^\beta_\mathrm{DM}(\rho)$ is non-empty. This implies that $\rho$ is a minimizer for the variational principle in Eq.~\eqref{eq: Omega-from-F}, where the potential is in the negative subdifferential, i.e., $v \in -\partial F^\beta_\mathrm{DM}(\rho)$. The minimizer of Eq.~\eqref{eq:F_DM-const-search}, guaranteed to exist by Theorem \ref{th: existence of the minimizer for the universal functional}, is equal to the Gibbs state $\Gamma_v$. We conclude, that $\rho$ is  $v$-representable. On the other hand, every density of a Gibbs state $\Gamma_v$ with $v \in H^{-1}(\T)$ is contained in $\densset$ according to Proposition \ref{prop: Gibbs state density in densset} and is strictly positive by Theorem \ref{th: Gibbs state density non-zero}. In particular, every $v$-representable density is in $\vrepset$ and $F^\beta_\mathrm{DM}$ is not differentiable at densities in $\densset\setminus\vrepset$. On the other hand, Theorem \ref{th: differentiability of universal functional} shows that the thermal universal functional $F^\beta_\mathrm{DM}$ is Gâteaux differentiable at any $\rho \in \vrepset$.
\end{proof}

In this work we provided an explicit characterization of the set of uniquely $v$-representable densities on the one-dimensional torus at elevated temperatures and established (Gâteaux) differentiability of the thermal universal functional. In the same manner the theory can be developed for any other interval with Neumann boundary conditions. As an extensive parameter, the volume needs to be finite, and thus a periodic setting is the natural choice for elevated temperatures. The restriction to one dimension remains as the main limitation of the current formulation even though some of the results like the uniqueness of the Gibbs state, strict positivity of the density and the uniqueness of the potential also hold in higher dimensions. Note that we cannot conclude Gâteaux differentiability in higher dimensions since continuity of the thermal universal functional might not be guaranteed. Another interesting observation is that we need to include distributional potentials for achieving a $v$-representability result already in one dimension. In DFT often the potential set $L^{3/2}+L^\infty$ is considered~\cite{Lieb1983}. However, for zero temperature and single particles it was already observed by \citet{ENGLISCH1983} and \citet{CCR1985} that certain densities require distributional potentials for their realization. \citet{sutter2024solution} showed that distributional potentials of the form Eq.~\eqref{eq:theorem1-b} are sufficient for $v$-representability of the set $\vrepset$. For elevated temperature we further show that such distributional potentials are also necessary. Since it cannot be expected that in higher dimensions the space of necessary potentials gets more regular, we conjecture that distributional potentials (probably of a higher class) are also required there if a similar proof of $v$-representability is feasible. However, in the current construction, the continuous Sobolev embedding $H^1 \hookrightarrow L^\infty$ is crucial and it holds in one dimension. For higher dimensions, we either need to require higher regularity or higher integrability conditions. This means that for densities we consider the Sobolev space $W^{m,p}$ and need to have $m$ and/or $p$ large enough. The problem is then that not all Gibbs state densities might be contained in $W^{m,p}$ and that a Hamiltonian with a potential in $W^{-m,q}=(W^{m,p})^\ast$ will usually not be self-adjoint~\cite{HERCZYNSKI1989-KLMN}. Another open question is how to extend the results to the grand canonical ensemble of variable particle number. In this case, the particle-number operator is not bounded anymore and this is currently needed in the proof of Theorem~\ref{th: existence of the minimizer for the universal functional}, where we show that the weak-* limit yields the same density as the sequence. Additionally, we do not know if the square root of the Gibbs state density is in $H^1$. However, we think that the rest of the theory does not fundamentally change for the grand canonical ensemble and that we retrieve the same conclusions as for the canonical ensemble. 

\begin{acknowledgments}
We are grateful to Thiago Carvalho Corso for helpful feedback on a previous version of this article.
MP acknowledges support from the German Research Foundation under Grant SCHI 1476/1-1 and from ERC-2021-STG grant agreement No.\ 101041487 REGAL. The authors SMS and KJHG thank The Netherlands Organization for
Scientific Research, NWO, for its financial support under
Grant No.~OCENW.KLEIN.434 and Vici Grant No.~724.017.001. RvL acknowledges the Academy of Finland grant under project number 356906. MR acknowledges the Cluster of Excellence ``CUI: Advanced Imaging of Matter'' of the Deutsche Forschungsgemeinschaft (DFG), EXC 2056, project ID 390715994.
\end{acknowledgments}

\begin{appendix}
\section{Proof of Proposition \ref{prop: minimizer of Helmhotz functional/grand potential}}\label{app: proof minimizer properties}
We still need to show that Proposition \ref{prop: minimizer of Helmhotz functional/grand potential} holds. The result that the unique minimizer of the Helmholtz free energy is given by the Gibbs state $\Gamma_v=\e^{\beta H_v}/Z(v)$ and the relation $\Omega^\beta_v(\Gamma_v)=-\beta \log Z(v)$ is shown in standard literature. We give here a more elaborated version than our previous works~\cite{GiesbertzRuggenthaler2019, SutterGiesbertz2023} in the form of a formal proof without going into mathematical subtleties. A classical proof can be found in \citet[Appendix]{mermin1965thermal} for the grand canonical ensemble but works analogously for the canonical one.

\begin{proof}[Proof of Proposition \ref{prop: minimizer of Helmhotz functional/grand potential}]
To determine the minimizer of the Helmholtz free energy $\Omega^\beta_v(\Gamma)= \Tr \{\Gamma (H_v+\beta^{-1}\log\Gamma)\}$ we make it stationary with respect to small perturbations $\delta \Gamma$ that obey $\Tr \{\delta \Gamma\}=0$. We compute
\begin{equation}
\label{minimizer Helmholtz functional: perturbation difference}
    \Omega^\beta_v(\Gamma+\epsilon \delta \Gamma)-\Omega^\beta_v(\Gamma)=\epsilon \Tr \{\delta \Gamma H_v\}+\beta^{-1}\Tr \{(\Gamma + \epsilon \delta \Gamma)\log(\Gamma + \epsilon \delta \Gamma)\}- \beta^{-1}\Tr \{\Gamma \log \Gamma\}.
\end{equation}
We use the Taylor expansion of $x\log(x)$ to get $(\Gamma + \epsilon \delta \Gamma)\log(\Gamma+\epsilon \delta \Gamma)= \Gamma\log \Gamma + \log (\Gamma) \epsilon \delta \Gamma + \epsilon \delta \Gamma+ \mathcal{O}(\epsilon^2)$. Dividing by $\epsilon$ turns Eq.~\eqref{minimizer Helmholtz functional: perturbation difference} into 
\begin{equation}\begin{aligned}
\label{eq: differential quotient for Helmholtz functional}
    \frac{\Omega^\beta_v(\Gamma+\epsilon \delta \Gamma)-\Omega^\beta_v(\Gamma)}{\epsilon}&=\Tr \{\delta \Gamma H_v\}+ \beta^{-1}\Tr \{\delta \Gamma \log \Gamma\} + \beta^{-1}\Tr \{\delta \Gamma\}+\mathcal{O}(\epsilon)\\ 
    &= \Tr \{\delta \Gamma H_v\}+ \beta^{-1}\Tr \{\delta \Gamma \log \Gamma\} +\mathcal{O}(\epsilon).
\end{aligned}\end{equation}
The quotient in Eq.~\eqref{eq: differential quotient for Helmholtz functional} must go to zero for any minimizer $\Gamma_v$ as $\epsilon \to 0$. Hence we have $\Tr \{\delta \Gamma H_v\}+ \beta^{-1}\Tr \{\delta \Gamma \log \Gamma_v\}=\Tr \{\delta \Gamma(H_v+\beta^{-1}\log \Gamma_v)\}=0$ for all $\delta\Gamma$ which requires $H_v+\beta^{-1}\log \Gamma_v=\mathrm{const}$. This then leads to the unique minimizer $\Gamma_v=\e^{-\beta H_v}/Z(v)$ where $Z(v)$ is determined from the normalization constraint.\\
To get the expression for the Helmholtz free energy, let us first assume that $Z(v)$ is finite, i.e., that the Gibbs state exists. Using the explicit expression of the Gibbs state $\Gamma_v$ and $\log \Gamma_v=-\beta H_v- \log Z(v)$ gives the expression
\begin{equation}
    \Omega^\beta(v)=\Omega^\beta_v(\Gamma_v)=\Tr \{\Gamma_v(H_v+\beta^{-1}\log \Gamma_v)\}=\Tr \{\Gamma_v(H_v-H_v-\beta^{-1}\log Z(v))\}=-\beta^{-1}\log Z(v).
\end{equation}
For the last remaining part of the proposition we need to show that $Z(v)$ is indeed finite for any potential $v \in H^{-1}(\T)$. First, we show that the partition function is finite for a non-interacting system with zero potential, i.e., $H_v=T$. We denote this partition function by $Z_0(0)$. The eigenstates of the kinetic-energy operator $T=-\frac{1}{2}\Delta$ on $\T$ are Slater determinants built from the functions $\e^{\i 2 \pi p\cdot x}$ with $p \in \Z$ which have eigenvalues $2\pi^2 p^2$. We estimate
\begin{equation}
\label{eq: finite Z_0(0)}
    Z_0(0)=\Tr \{\e^{-\beta T}\}\leq\sum \limits_{p_1, \ldots, p_N \in \Z}\e^{-\beta 2\pi^2(p_1^2+\ldots+p_N^2)}=\Big[\sum \limits_{p \in \Z}\e^{-\beta 2\pi^2p^2}\Big]^N<\infty.
\end{equation}
To show that $Z(v)$ is finite for a general interaction and a general potential we first show that $\Omega^\beta(v)$ is finite. For that recall that we assume that the interaction satisfies the KLMN condition from Eq.~\eqref{eq:KLMN condition for the interaction}. Moreover, \citet[Lem.~16]{sutter2024solution} showed that any potential $v\in H^{-1}$ also fulfills the KLMN condition $|V(\Psi, \Psi)| \leq a' T(\Psi)+b'$ where $a' > 0$ can be made arbitrary small. Here $V(\cdot, \cdot)$ is the quadratic form produced by $v \in H^{-1}$. Combing these two KLMN conditions guarantees the existence of new constants $0\leq a< 1$ and $b\geq 0$ such that $\langle \Psi, H_v \Psi \rangle  \geq (1-a)T(\Psi)-b$ for all $\Psi \in \psiset$. This condition translates to density matrices and we get
\begin{equation}\label{eq: bound and finiteness of Helmholtz free energy}
    \Omega^\beta(v)=\inf_{\Gamma\in\DMset} \Tr \{\Gamma(H_v+\beta^{-1} \log(\Gamma))\}\geq (1-a)\inf_{\Gamma\in\DMset} \Tr \{\Gamma(T+\big(\beta(1-a)\big)^{-1}\log(\Gamma))\}-b=-\beta^{-1}\log\big(Z_0(0)\big)-b>-\infty.
\end{equation}
Let $H_v=\sum _{j=1}^\infty \lambda_j \ket{\psi_j}\bra{\psi_j}$ be the spectral decomposition of the Hamiltonian $H_v$ which exist due to the compact embedding of $Q(H_v)=\big(H_{\C}^1(\T)\otimes \C^2\big)^{\wedge N}$ into $\big(L^2_{\C}(\T)\otimes \C^2\big)^{\wedge N}$. We define a sequence of density matrices by 
\begin{equation}
    \Gamma_n=\frac{1}{Z_n[v]}\sum \limits_{j=1}^n \e^{-\beta \lambda_j}\ket{\psi_j}\bra{\psi_j} \qquad \text{with} \qquad Z_n[v]=\sum \limits_{j=1}^n  \e^{-\beta \lambda_j}
\end{equation}
and we compute
\begin{equation}
    \Omega^\beta(v) < \Tr \{\Gamma_n(H_v+\beta^{-1}\log(\Gamma_n))=\frac{1}{Z_n[v]}\sum \limits_{j=1}^n \e^{-\beta \lambda_j}(-\beta^{-1}\log Z_n[v])=-\beta^{-1}\log Z_n[v]. 
\end{equation}
The left-hand side is bounded from below by Eq.~\eqref{eq: bound and finiteness of Helmholtz free energy} and we get that $Z_n[v]$ is uniformly bounded from above. This implies $Z(v)<\infty$.\\
Finally, to show that $\Gamma_v\in\DMset$ we need to show $S(\Gamma_v)<\infty$. We use again $\log \Gamma_v=-\beta H_v- \log Z(v)$ and show with Lemma~\ref{lem: finite kinetic energy of Gibbs state} that
\begin{equation}
    S(\Gamma_v) = -\Tr\{\Gamma_v\log\Gamma_v\} = \beta\Tr\{\Gamma_v H_v\} + \log Z(v)\Tr\{\Gamma_v\} <\infty.
\end{equation}
\end{proof}

\section{Proof of Proposition \ref{prop: Gibbs state density in densset}}\label{app: proof Gibbs state density in densset}
In this section we show that the density of a Gibbs state $\Gamma_v$ is contained in $\densset$. To do so, note that  as explained in the proof of Proposition \ref{prop: minimizer of Helmhotz functional/grand potential} we have for each potential $v$ two constants $0 \leq a <1$ and $b \geq 0$ such that 
\begin{align}
    \label{eq: lower and upper bound of Hamiltonian with T}
    (1-a)T-b\leq H_v \leq (1+a)T+b,
\end{align}
where $T=-\frac{1}{2}\Delta$ is the kinetic energy operator. Moreover, the Hamiltonian $H_v$ (and also $-\Delta/2$) has discrete spectrum due to the compact embedding $Q(H_v)=\bigl(H_{\C}^1(\T)\otimes \C^2\bigr)^{\wedge N} \hookrightarrow \bigl(L^2_\C(\T)\otimes \C^2\bigr)^{\wedge N}$. We can express $H_v$ and $T$ in their respective spectral decomposition
\begin{align}
\label{eq: spectral decomposition Hamiltonian and kinetic operator}
    H_v= \sum \limits_{j=1}^\infty \lambda_j \ket{\psi_j}\bra{\psi_j}, \quad T=\sum \limits_{j=1}^\infty \mu_j \ket{\phi_j}\bra{\phi_j},
\end{align}
where we additionally require $\lambda_1 \leq \lambda_2\leq \ldots$ and $\mu_1 \leq \mu_2 \leq \ldots$. Then we get the following lemma.
\begin{lemma}\label{lem: lower, upper bound for eigenvalues}
    The $n$'th eigenvalue of $H_v$ has the following lower and upper bound, $(1-a)\mu_n-b\leq \lambda_n\leq (1+a)\mu_n+b$.
\end{lemma}
\begin{proof}
   Let $\phi_j$ be the states introduced in Eq.~\eqref{eq: spectral decomposition Hamiltonian and kinetic operator}. By the min-max principle for self-adjoint operators and by Eq.~\eqref{eq: lower and upper bound of Hamiltonian with T} we get
   \begin{equation}
    \begin{aligned}
       \lambda_n&=\min_{\substack{\varphi_1, \ldots, \varphi_n\\
       \mathrm{orthon.}}} \max \big\{\bra{\psi}H_v\ket{\psi} \mid \psi \in \mathrm{span}\{\varphi_1, \ldots, \varphi_n\}, \norm{\psi}_2=1\big\}\leq \max \big\{\bra{\psi}H_v\ket{\psi}\mid \psi \in \mathrm{span}\{\phi_1, \ldots, \phi_n\}, \norm{\psi}_2=1\big\}\\
       &\leq \max \big\{(1+a)\bra{\psi}T\ket{\psi}+b \mid \psi \in \mathrm{span}\{\phi_1, \ldots, \phi_n\}, \norm{\psi}_2=1\big\}=(1+a)\mu_n+b.
    \end{aligned}
    \end{equation}
    For the lower bound we proceed in the same way but we use the max-min principle instead to get
    \begin{equation}
    \begin{aligned}
        \lambda_n&=\max_{\substack{\varphi_1, \ldots, \varphi_{n-1}\\ \mathrm{orthon.}}} \min \big\{\bra{\psi}H_v\ket{\psi}\mid \psi \perp \varphi_1, \ldots, \varphi_{n-1}, \norm{\psi}_2=1\big\}\geq \min \big\{\bra{\psi}H_v\ket{\psi}\mid \psi \perp \phi_1, \ldots, \phi_{n-1}, \norm{\psi}_2=1\big\}\\
        &\geq \min \big\{(1-a)\bra{\psi}T\ket{\psi}-b\mid \psi \perp \phi_1, \ldots,\phi_{n-1}, \norm{\psi}_2=1\big\}=(1-a)\mu_n-b.
    \end{aligned}
    \end{equation}
    This gives the desired bounds.
\end{proof}
We use the lower and upper bound of the eigenvalues to prove that expectation value of the kinetic energy operator is finite.
\begin{lemma}\label{lem: finite kinetic energy of Gibbs state}
    For all $v \in H^{-1}(\T)$ we have $0\leq \Tr\{T\Gamma_v\}<\infty$ and $-\infty<\Tr\{H_v\Gamma_v\}<\infty$.
\end{lemma}
\begin{proof}
   The Gibbs state is given by $\Gamma_v=\frac{1}{Z(v)}\sum _{j=1}^\infty \e^{-\beta \lambda_j}\ket{\psi_j}\bra{\psi_j}$. We first show that the expectation value is finite by using Lemma \ref{lem: lower, upper bound for eigenvalues}. That is,
   \begin{equation}
   \begin{aligned}
       -\infty &< \Tr\{H_v \Gamma_v\}=\frac{1}{Z(v)}\sum \limits_{j=1}^\infty \lambda_j \e^{-\beta \lambda_j}\leq \frac{1}{Z(v)}\sum \limits_{j=1}^\infty\big[(1+a)\mu_j+b\big]\e^{-\beta[(1-a)\mu_j-b]}\\
       &=\frac{\e^{\beta b}(1+a)}{Z(v)}\sum \limits_{j=1}^\infty\mu_j \e^{-\beta(1-a)\mu_j}+\frac{\e^{\beta b}b}{Z(v)} \sum \limits_{j=1}^\infty \e^{-\beta(1-a)\mu_j}. 
   \end{aligned}
   \end{equation}
   A similar computation as in Eq.~\eqref{eq: finite Z_0(0)} shows that the two sums are finite and we get $\Tr\{H_v \Gamma_v\}<\infty$. Using Eq.~\eqref{eq: lower and upper bound of Hamiltonian with T} and positivity of the kinetic energy operator gives
   \begin{align}
       -b\leq (1-a)\Tr\{T\Gamma_v\}-b\leq \Tr\{H_v\Gamma_v\}<\infty,
   \end{align}
   which concludes the proof.
\end{proof}
We can now prove Proposition \ref{prop: Gibbs state density in densset}.
\begin{proof}[Proof of Proposition \ref{prop: Gibbs state density in densset}]
    For simplicity we introduce $\alpha_j=\e^{-\beta \lambda_j}/Z(v)$ such that the Gibbs state is $\Gamma_v=\sum_{j=1}^\infty \alpha_j \ket{\psi_j}\bra{\psi_j}.$ Then, its density is given by $\rho=\sum_{j=1}^\infty \alpha_j \rho_j$ with $\psi_j \mapsto \rho_j$. By Theorem \ref{th:N-rep and properties of densset} all the pure state densities $\rho_j$ are contained in $\densset$. We define a sequence $\rho^n=\sum_{j=1}^n \alpha_j \rho_j$. Clearly, $\rho^n \to \rho$ in $L^1(\T)$ and $\norm{\sqrt{\rho^n}}_2^2\leq N$. Since $\int (\nabla \sqrt{\rho})^2\d x$ is convex on $\densset$~\cite[remark after Eq.~(1.10)]{Lieb1983}, we find by Jensen's inequality and by Theorem~\ref{th:N-rep and properties of densset} that
    \begin{align}
        \int (\nabla \sqrt{\rho^n})^2 \d x \leq\sum \limits_{j=1}^n \alpha_j \int (\nabla \sqrt{\rho_j})^2 \d x\leq \sum \limits_{j=1}^n \alpha_j 2T(\psi_j)\leq  2\Tr\{T \Gamma_v\}<\infty,
    \end{align}
    where we get a uniform bound by Lemma \ref{lem: finite kinetic energy of Gibbs state}. In particular, $\sqrt{\rho^n}$ is a bounded sequence in $H^1(\T)$ and by the Banach--Alaoglu theorem there is a subsequence (also denoted by $\sqrt{\rho^n}$) and $f \in H^1(\T)$ such that $\sqrt{\rho^n}\rightharpoonup f$ weakly in $H^1(\T)$. But since $H^1(\T)$ is compactly embedded in $L^2(\T)$ there is also a subsequence (again denoted by $\sqrt{\rho^n}$) such that $\sqrt{\rho^n}\to f$ strongly in $L^2(\T)$. Then, we have
    \begin{align}
       \norm{\rho^n-f^2}_1=\int (\sqrt{\rho^n}-f)(\sqrt{\rho^n}+f)\d x \leq \norm{\sqrt{\rho^n}-f}_2\norm{\sqrt{\rho^n}+f}_2 \to 0 
    \end{align}
    since $\norm{\sqrt{\rho^n}+f}_2$ is bounded. This means that $\rho^n \to f^2$ strongly in $L^1(\T)$ and we conclude $\sqrt{\rho}=f \in H^1(\T)$.
\end{proof}

\end{appendix}

\bibliographystyle{apsrev4-1-mod}
\bibliography{refs}
\end{document}